\documentclass[12pt]{article}
\usepackage{amsmath}
\usepackage{amsthm}
\usepackage{amssymb}
\usepackage{mathrsfs}
\usepackage{authblk}
\usepackage[usenames]{color}
\usepackage[latin5]{inputenc}
\usepackage{cite}
\usepackage{pdfsync}
\usepackage{enumitem}
\setlist{parsep=0pt,itemindent=0pt}

\theoremstyle{definition}

\newtheorem{thm}{Theorem}

\newtheorem{rmk}{Remark}

\newtheorem{exa}{Example}
\numberwithin{equation}{section}
\numberwithin{thm}{section}
\numberwithin{lemma}{section}
\numberwithin{prop}{section}
\numberwithin{cor}{section}
\numberwithin{rmk}{section}
\numberwithin{defn}{section}
\numberwithin{exa}{section}

\newcommand{\gen}[1]{\partial_{#1}}

\newcommand{\lie}{\mathfrak g}

\DeclareMathOperator{\Sl}{sl}

\begin{document}
\pagenumbering{arabic}
\clearpage
\thispagestyle{empty}

\title{Linearizability for third order evolution equations}

\author[1]{P.~Basarab-Horwath\thanks{peter.basarab-horwath@liu.se}}
\author[2]{F.~G\"ung\"or\thanks{gungorf@itu.edu.tr}}

\affil[1]{Department of Mathematics, Link\"oping University,  S-581 83
Link\"oping, Sweden}
\affil[2]{Department of Mathematics, Faculty of Science and Letters, Istanbul Technical University, 34469 Istanbul, Turkey}

\date{}

\maketitle



\begin{abstract}
The problem of linearization for third order evolution equations is considered. Criteria for testing equations for linearity are presented. A class of linearizable equations depending on arbitrary functions is obtained by requiring presence of an infinite-dimensional symmetry group. Linearizing transformations for this class are  found using  symmetry structure and local conservation laws. A number of special cases as examples are discussed. Their transformation to equations within the same class by differential substitutions and connection with KdV and mKdV equations are also reviewed in this framework.

\end{abstract}

\section{Introduction}
In a series of papers \cite{GungorLahnoZhdanov2004, Basarab-HorwathGuengoerLahno2013, Basarab-HorwathGuengoerLahno2013a} we investigated the group classification problem for third order nonlinear evolution equations. In \cite{GungorLahnoZhdanov2004} the main focus was on classification of the lower-dimensional Lie point symmetry algebras $\lie$ of dimension $\leq 4$ for equations of the form
\begin{equation}\label{main-F}
u_t=u_{xxx}+F(t,x,u,u_x,u_{xx}).
\end{equation}
The Ref. \cite{Basarab-HorwathGuengoerLahno2013, Basarab-HorwathGuengoerLahno2013a} were devoted to equations of the class
\begin{equation}\label{main-F-G}
u_t=F(t,x,u,u_x,u_{xx})u_{xxx}+G(t,x,u,u_x,u_{xx}), \quad F\not=0.
\end{equation}
The classification for \eqref{main-F-G} for semi-simple Lie point symmetry algebras was carried out in \cite{Basarab-HorwathGuengoerLahno2013}, and for solvable algebras in a separate work \cite{Basarab-HorwathGuengoerLahno2013a}. In these two papers only non-linear equations were listed and linear(izable) equations were identified using  the following criterion for linearizability: if an equation admits as a symmetry algebra an abelian algebra of dimension four, or a rank-one realization of a three-dimensional Lie algebra, then it is linearizable  (for details, see \cite{Basarab-HorwathGuengoerLahno2013}). For solvable algebras we give a similar test which can be used  to eliminate linear equations in the course of classification. By a linearizable equation we mean  one which is either obviously linear in $u, u_x, u_{xx}, u_{xxx}$ or one which is transformable to such a linear equation by an admissible equivalence transformation.

An alternative to using these types of theorems to identify the symmetry structure of linearizable equations we can of course test  a given equation for linearizability by looking at the maximum symmetry group: if the equation is linearizable then it admits an infinite (point) symmetry algebra.
The existence of an infinite-dimensional (point) symmetry algebra can still serve as a criterion of linearizability by point transformation, but this is only practical if the Lie algorithm is easy to apply. Conversely, we can generate a family of linearizable equations by requiring  the equation at hand to be invariant under some given infinite-dimensional subalgebra. This is done in section \ref{lin-evol-eqs} of this article. Once we know the symmetry algebra that contains the infinite-dimensional ideal of the full algebra, it is easy to find the corresponding linearizing transformation by transforming the vector field to one manifesting a linear superposition law for some linear KdV equation.

In section \ref{con-law} we construct conservation laws for a class of nonlinear evolution equations which are linearizable for a special choice of the coefficients. Using these laws we investigated the linearizability of such a class.

This paper is organized as follows: In section \ref{g-class-LE} we revisit the group classification problem for linear evolution equations. In section \ref{linear-crit} we state two  theorems for equations with abelian symmetry algebras of dimension $\geq 3$ and rank-one solvable algebras of dimension three as linearizability criteria. In section \ref{lin-evol-eqs} we derive a class of linearizable equations using infinitesimal Lie point symmetries and conservation laws. We illustrate our results with a number of examples.

\section{Classification of linear equations}\label{g-class-LE}
In \cite{GungorLahnoZhdanov2004} we showed that the most general linear third order PDE
\begin{equation} \label{linear-general}
u_t = f_1 (x,t) u_{xxx} +f_2(x,t) u_{xx} +f_3(x,t) u_x +f_4(x,t) u
+f_5(x,t).
\end{equation}
can be transformed to either
\begin{equation}  \label{linear-normal}
u_t = u_{xxx} +A(t,x) u_x +B(t,x) u,
\end{equation}
or
\begin{equation}  \label{linear-normal-2}
u_t = u_{xxx} +A(t,x) u_{xx} +B(t,x) u_x,
\end{equation}
where $A, B$ are arbitrary smooth functions of $t$ and $x$,
by the equivalence transformation
\begin{equation}\label{trans}
\tilde t = t,\quad \tilde x = X(t,x), \quad u = U_1(t,x) \tilde{u}(\tilde t,
\tilde x) +U_0(t,x), \quad U_1 \not =0, \ X_x \not =0.
\end{equation}
\begin{table}\caption{Symmetry Classification of \eqref{linear-normal}}\label{tab1}
\begin{center}
\vskip 3mm
\begin{tabular}{|c|c|c|l|}
\hline
N   & $A$ & $B$ & Generators \\[2mm] \hline
1& $A(x)$ &  $B(x)$ &$\partial_t$ \\[2mm] \hline &&& \\
2& $0$ & $\dot f (t)x$
&$\partial_x+f(t) u \partial_u, \ \ddot f \not =0$
\\[2mm] \hline &&& \\
3&  $m x^{-2}, \ m \in \mathbb{R}$ & $n x^{-3}, \ n
\in \mathbb{R},$ &$\partial_t, t \partial_t +\frac{1}{3} x
\partial_x$
\\[2mm] & & $ \ |m|+|n| \not =0$ & \\[2mm] \hline &&& \\
4&  $0$ &
$\varepsilon x, \ \varepsilon = \pm 1$ &$\partial_t,  \partial_x
+\varepsilon t u
\partial_u$ \\[2mm] \hline &&& \\
5&  $0$ & $-m
t^{-\frac{4}{3}}x,$ \   &$\partial_x+3m
t^{-\frac{1}{3}}u\partial_u,$ \\[2mm] & &$\ m \in \mathbb{R}, \ m \not
=0$&$t \partial_t +\frac{1}{3} x \partial_x$ \\[2mm] \hline
\end{tabular}
\end{center}
\end{table}
The same problem for $n=4$ was undertaken in \cite{HuangQuZhdanov2012}. In a very recent work, an extension to $n+1$-dimensional case for $n>2$  of the group classification problem for linear evolution equations has been carried out in \cite{BihloPopovych2016} where the authors solved the problem in full generality and compared their results with those of $n=3$ of \cite{GungorLahnoZhdanov2004} (also for $n=4$ of \cite{HuangQuZhdanov2012}) commented that there is a redundant  case, where $A(x)=a\in \mathbb{R}$, $B(x)=0$ with generators $\gen t, \gen x, t\gen t+1/3(x-2at)\gen x$ corresponding to $N=6$ in \cite{GungorLahnoZhdanov2004}), that can be removed by an equivalence transformation (in fact a simple Galilei transformation $(t,x,u)\to (t,x+a t,u)$ is sufficient to set $a=0$) and an omission of a constraint on one of the coefficients of the generators. Above, in Table \ref{tab1}, we reproduce the group classification table with the necessary corrections made.

\section{Linearizability and Linearization Criteria}\label{linear-crit}
In this Section we state two theorems establishing the relation between the structure of symmetry algebra and linearization.

For the sake of completeness we quote the following Theorem from \cite{Basarab-HorwathGuengoerLahno2013}.

\begin{thm}\label{abeliansymm} All inequivalent, admissible abelian Lie algebras of vector fields of the form $Q=a(t)\gen t + b(t,x,u)\gen x + c(t,x,u)\gen u$ are given as follows:
\begin{align*}
& A=\langle \gen t\rangle,\quad A=\langle \gen u \rangle\;\; (\dim A=1)\\
& A=\langle \gen t, \gen u \rangle,\quad A=\langle \gen x, \gen u \rangle,\quad A=\langle \gen u, x\gen u\rangle\;\; (\dim A=2)\\
& A=\langle \gen t, \gen x, \gen u\rangle,\quad A=\langle \gen t, \gen u, x\gen u \rangle,\quad A=\langle \gen u, x\gen u, c(t,x)\gen u \rangle,\; c_{xx}\neq 0\;\;\\
& (\dim A=3)\\
& A=\langle \gen t, \gen u, x\gen u, c(x)\gen u \rangle,\; c''(x)\neq 0\;\, (\dim A=4)\\
& A=\langle \gen u, x\gen u, c(t,x)\gen u, q_1(t,x)\gen u,\dots, q_k(t,x)\gen u \rangle,\;\;(\dim \mathsf{A}\geq 4)\\
& c_{xx}\neq 0,\, (q_1)_{xx}\neq 0,\dots, (q_k)_{xx}\neq 0.
\end{align*}
\end{thm}


We also have the following result, which was central to identifying realizations of solvable algebras which lead to linear(izable) equations. It allowed us to omit many of the cases of 5-dimensional solvable Lie algebras in Mubarkzianov's list \cite{Basarab-HorwathGuengoerLahno2013a}.

\begin{thm}\label{linearizableeqns2}
If the evolution equation \eqref{main-F-G} admits a rank-one solvable Lie algebra $\mathsf A=\langle e_1,e_2,e_3\rangle$ of dimension three or an abelian Lie algebra $\mathsf{A}$ with $\dim\mathsf{A}\geq 4$ as symmetries then the equation is linearizable.
\end{thm}

\begin{proof}
If the algebra $\mathsf{A}$ is abelian then the Theorem \eqref{abeliansymm} gives the result.

If $\mathsf{A}$ is not abelian then we may assume a chain of ideals
$$\langle e_1\rangle \lhd \langle e_1, e_2\rangle  \lhd  \langle e_1,e_2,e_3\rangle.$$
Now, $\langle e_1, e_2\rangle $ is either an abelian or it is a non-abelian solvable algebra.  If $\langle e_1, e_2\rangle $ is nonabelian, then we may take  $[e_1,e_2]=e_1$ and in canonical form  we have $\langle e_1, e_2\rangle =\langle \gen u, u\gen u\rangle $ since the realization is to be rank-one (this is the canonical form of the rank-one realization of the solvable Lie algebra $\mathsf{A}_{2.2}$). We then  take $e_3=c(t,x,u)\gen u$ so that
$$[e_1,e_3]=\alpha e_1+ \beta e_2,  \quad [e_2,e_3]=\gamma e_1+ \delta e_2.$$
Elementary calculations give $c(t,x,u)\gen u=\alpha e_2+\kappa e_1+\gamma(t,x)\gen u$ so that we have $\mathsf{A}=\langle\gen u, u\gen u, \gamma(t,x)\gen u\rangle$. If $\gamma_x=0$ then we have $e_3=\gamma(t)\gen u$ and this together with $\gen u$ gives $\dot{\gamma}=0$ in the determining equation for $G$ of \eqref{main-F-G} so that we must have $\gamma_x\neq 0$ for $\dim \mathsf{A}=3$. In this case, we use the residual equivalence group
$$\mathscr{E}(e_1,e_2): t'=T(t), \quad x'=X(t,x), \quad  u'=u, \qquad X_x \ne 0.$$
Under such a transformation $\gamma(t,x)\gen u$ is mapped to $\gamma(t,x)\gen {u'}$. Since $\gamma_x\ne 0$ we may take $x'=X(t,x)=\gamma(t,x)$ so that we have  $e_3=x\gen u$ in canonical form. Hence, $\mathsf{A}=\langle\gen u, u\gen u, x\gen u\rangle$ and this linearizes the equation.

Then we take $\langle e_1, e_2\rangle$ as the abelian ideal $ \langle \gen u, x\gen u \rangle $ in canonical form. Putting $e_3=c(t,x,u)\gen u$, the same sort of reasoning as above leads to $e_3=(\alpha u+\gamma(t,x))\gen u$. We can set $\gamma\to 0$ by the residual equivalence group $\mathscr{E}(e_1,e_2): t'=T(t), x'=x, u'=u+U(t,x)$ so that
we  again have $\mathsf{A}=\langle\gen u, x\gen u, u\gen u\rangle$.

\end{proof}

\begin{exa}
$$u_t=u_x^{-3}u_{xxx}-3u_x^{-4}u_{xx}^2$$
It admits a rank-one solvable subalgebra of dimension three
$\langle\gen x, x\gen x, u\gen x\rangle$. So by the above theorem  linearizable: The subalgebra is transformed to the canonical form $\mathsf{A}=\langle\gen u, u\gen u, x\gen u\rangle$ by the hodograph transformation (an equivalence transformation) $$(t,x,u)\to (t,u,x)$$ with the following induced transformations of the derivatives:
$$u_t\to -\frac{u_t}{u_x},  \quad u_x\to \frac{1}{u_x},  \quad u_{xx}\to \frac{u_{xx}}{u_x^3},  \quad u_{xxx}\to -\frac{u_{xxx}}{u_x^4}+3\frac{u_{xx}^2}{u_{x}^5},$$
the equation linearizes to $u_t=u_{xxx}$.

The full symmetry algebra is spanned by the infinite-dimensional algebra
$$\langle\gen t, \gen x, x\gen x, t\gen t+\frac{u}{3}\gen u\rangle \vartriangleright \langle\mathbf{v}(\gamma)\rangle,$$ where
$$\mathbf{v}(\gamma)=\gamma(t,u)\gen x,  \quad \gamma_t=\gamma_{uuu}.$$
The presence of the abelian ideal $\mathbf{v}(\gamma)$ suggests the hodograph transformation.
\end{exa}

These theorems are useful in checking the linearity of a given equation with a special symmetry group. However, even if the linearity test is passed it might not be always possible to find a linearizing transformation without knowing the maximal symmetry group of the equation.

\section{A class of linearizable evolution equations}\label{lin-evol-eqs}
We consider a subclass of \eqref{main-F} where $F$ is a polynomial in $u_x^3, u_x^2, u_{xx}, u_{x}$
\begin{equation}\label{class}
  F=3f(u)u_xu_{xx}+g(u)u_x^3+h(u)u_x^2+p(u)u_{xx}+q(u)u_x,
\end{equation}
where $f, g, h, p, q$ are arbitrary smooth functions of their arguments. We shall investigate linearizability of $u_t=u_{xxx}+F$ from different aspects.

\subsection{Presence of infinite-dimensional symmetry algebra}
We know from \cite{GungorLahnoZhdanov2004} that the Lie symmetry algebra is generated by the vector field
\begin{equation}\label{gvf}
X = \tau (t)\partial _t  + (\frac{{\dot \tau }} {3}x + \rho
(t))\partial _x  + \phi (t,x,u)\partial _u,
\end{equation}
where $\tau(t)$, $\rho(t)$,  $\phi(t,x,u)$ and $F$
satisfy a single determining equation (see Proposition 2.1 of \cite{GungorLahnoZhdanov2004}). Eq. \eqref{class} is invariant under the abelian subalgebra $\langle \partial_t, \partial_x\rangle$. We want to force  Eq.\eqref{class} to be invariant under at least an infinite-dimensional subalgebra depending on an arbitrary function $\varrho(t,x)$ satisfying a linear equation in normal form \eqref{linear-normal-2}.

We substitute $F$ into the determining equation and split it with respect to the linearly independent derivatives  $u_x^3, u_x^2, u_{xx}, u_{x}$ and then choose $\tau=0$, $\rho=0$ and $\phi=\varrho(t,x)k(u)$ in \eqref{gvf}, where $k(u)$ is to be determined in such a way that $\varrho(t,x)$  remains arbitrary
and satisfies  \eqref{linear-normal-2}. This requirement gives us the following relations for the coefficients of \eqref{class} and $k(u)$:
\begin{eqnarray}\label{deteqs}
   &&k'+fk=0,  \quad h=pf, \\
   &&k''+2fk'+gk=0,  \quad k'''+3fk''+2gk'+kg'=0,\\
   &&f'+f^2-g=0,  \quad f''-g'+2fg-2f^3=0,
\end{eqnarray}
where $p$ is a constant. It also follows that $q$ must be constant. Note that the equation $f''-g'+2fg-2f^3=0$ is just the first differential consequence of $f'+f^2-g=0$.
The consistency of the above system involving $f, g, k$ and derivatives boils down to a single condition
\begin{equation}\label{Riccati}
  f'+ f^2=g
\end{equation}
and $k$ and $h$ are determined from the first set of ODEs of   \eqref{deteqs}. We note that the equations $(4.5)$ also appeared in \cite{SandersWang1998} as the condition of existence of a Recursion operator for the corresponding system. The fact that the conditions $$f'+f^2-g=0,  \quad f''-g'+2fg-2f^3=0$$
are in fact just equation \eqref{Riccati} was omitted there. Throughout we shall assume $f\not=0$.  For a given $g$, Eq. \eqref{Riccati} can be regarded as a Riccati equation. The transformation $f=\sigma'/\sigma$ linearizes  \eqref{Riccati} to $\sigma''-g(u)\sigma=0$. Thus, for arbitrary $f$ we can generate a family of linearizable class of PDEs depending on $f(u)$. The symmetry vector field now has the form
\begin{equation}\label{inf-vf}
\mathbf{v}(\varrho)=\varrho(t,x)k(u)\gen u,  \quad k(u)=e^{-\int f(u)du},
\end{equation}
where $\varrho$ satisfies the linear equation
\begin{equation}  \label{linear-normal-3}
\varrho_t = \varrho_{xxx} +a \varrho_{xx} +b \varrho_x,
\end{equation}
where $a$, $b$ are arbitrary constants. Note that $b$ can be set to zero by the Galilean transformation $(t,x,u)\to (t,x+at,u)$ so there is no loss of generality in putting $b=0$ in the subsequent analysis.   We shall call \eqref{Riccati} a linearizability condition. Under this condition, the transformation
\begin{equation}\label{lin-trans}
w=\int\left(e^{\int f(u)du}\right)du=\int\sigma(u)du
\end{equation}
takes \eqref{inf-vf} to $\mathbf{\tilde{v}}(\varrho)=\varrho(t,x)\gen w$ and Eq. \eqref{class} to the linear equation
\begin{equation}\label{linear}
  w_t=w_{xxx}+aw_{xx}.
\end{equation}
The corresponding PDE has the form
\begin{equation}\label{class-f-g}
  u_t=u_{xxx}+3f(u)u_xu_{xx}+g(u)u_x^3+a(f(u)u_x^2+u_{xx})
\end{equation}
with the constraint $f'+f^2-g=0$ satisfied.
In particular, if $g=0$ we find $f(u)=1/(u-C)$, $C$ a constant and the linearizing transformation is given by $u=\pm\sqrt{2w}+C$.
\begin{rmk}
  Alternatively, we may look for a transformation $u=\Phi(w)$, $\Phi'\ne 0$. Substituting into \eqref{class}  we find
$$w_t=w_{xxx}+3[\frac{\Phi''}{\Phi'}+f \Phi']w_x w_{xx}+[\frac{\Phi'''}{\Phi'}+3f{\Phi''}+g{\Phi'}^2]w_x^3+[p\frac{\Phi''}{\Phi'}+h \Phi']w_x^2+pw_{xx}+qw_x.$$
For linearity we choose $(p,q)=(a,b)$ being constants and
$$\frac{\Phi''}{\Phi'}+f \Phi'=0,  \quad p\frac{\Phi''}{\Phi'}+h \Phi'=0,  \quad  \frac{\Phi'''}{\Phi'}+3f{\Phi''}+g{\Phi'}^2=0.$$
It is useful to express the derivatives of $\Phi$ in terms of its inverse function $w=\hat{\Phi}(u)=\Phi^{-1}(u)$ using $\Phi'(w)=1/\hat{\Phi}'(u)$. From the first two relations we find $h=pf$. The first relation can be written as   $\hat{\Phi}''(u)/\hat{\Phi}'(u)=f(u)$. Integration of this relation gives the linearizing transformation \eqref{lin-trans}. Differentiating the first relation and eliminating derivatives of $\Phi$  provides us the linearizability condition $f'+f^2=g$.
The resulting linear equation is \eqref{linear} for $b=0$.
\end{rmk}

Thus we have the following result:

\begin{thm}\label{one-family}
  The class of equations depending on an arbitrary function $f(u)$ and an arbitrary constant $a$
\begin{equation}\label{subclass}
  u_t=u_{xxx}+3f(u)u_xu_{xx}+(f'(u)+f(u)^2)u_x^3+a(f(u)u_x^2+u_{xx}),
\end{equation}
is linearizable (or S-integrable) by a change of dependent variable and the linearizing transformation is given by  \eqref{lin-trans}. The recursion operator for \eqref{subclass} \cite{SandersWang1998} is
\begin{equation}\label{recursion}
  \mathfrak{R}=(D_x+f(u)u_x)(D_x+f(u)u_x)=D_x^2+2f(u)u_xD_x+(f'+f^2)u_x^2+f(u)u_{xx}.
\end{equation}
\end{thm}

The three-parameter family of  equations
\begin{equation}\label{special-pde}
 u_t=u_{xxx}+\frac{3\alpha}{u}u_xu_{xx}+\frac{\beta}{u^2}u_x^3+a\left(\frac{\alpha}{u}u_x^2+u_{xx}\right),
\end{equation}
invariant under $\langle\gen t, \; \gen x, \;  u\gen u\rangle$
deserves a special attention.

There are two cases.

\par I.)
If $\beta=\alpha(\alpha-1)$ the linearity condition \eqref{Riccati} is satisfied.
The corresponding infinite dimensional symmetry algebra is generated by
$$\langle\gen t, \; \gen x, \; 3t\gen t+x\gen x, \; u\gen u\rangle \vartriangleright \langle \varrho(t,x) u^{-\alpha}\gen u\rangle,$$ where $\varrho$ solves $\varrho_t = \varrho_{xxx}+a\varrho_{xx}$.
For $\alpha\ne -1$ the transformation formula \eqref{lin-trans} gives the linearizing transformation $u=w^{1/(\alpha+1)}$ by scale invariance in $u$.

When $\alpha=-1$ ($\beta=2$) we have $w=\ln u$ (or $u=e^w$) as the linearizing transformation.

For $\alpha=\beta=2$ the linearizing transformation for \eqref{special-pde} is $w=u^3$. The second possibility $w=\ln u$ maps it to a linearizable equation in $w$ within the  class \eqref{class-f-g} with $f=3$, $g=9$. The corresponding linearizing transformation is $w=1/3\ln(3\tilde{w})$ which results in the previous transformation $w=u^3$.

Applying the transformation $u=e^w$  to \eqref{special-pde} and taking $\beta=\alpha(\alpha-1)$ we find a one-parameter family of equations  in the class \eqref{subclass} where $f=\alpha+1$
\begin{equation}\label{w-invar-pde}
  w_t=w_{xxx}+3(\alpha+1)w_xw_{xx}+(\alpha+1)^2w_x^3+a[(\alpha+1)w_x^2+w_{xx}].
\end{equation}
Eq. \eqref{w-invar-pde} is invariant under the infinite-dimensional symmetry algebra generated by the vector fields
\begin{equation}\label{inf-span}
  \langle\gen t, \gen x, 3t\gen t+x\gen x, \gen w,  e^{-(\alpha+1)w} \varrho(t,x)\gen w\rangle,
\end{equation}
where
$\varrho$ is an arbitrary solution of the linear equation
$\varrho_t=\varrho_{xxx}.$
Differentiating \eqref{w-invar-pde} with respect to $x$ and writing $\phi=w_x=u_x/u$ gives us
\begin{equation}\label{gen-STO}
  \phi_t=\phi_{xxx}+\frac{3(\alpha+1)}{2}(\phi^2)_{xx}+(\alpha+1)^2(\phi^3)_x
  +a[(\alpha+1)\phi^2+\phi_x]_x.
\end{equation}
For $\alpha=0$, $a=0$, Eq. \eqref{gen-STO}  is known as the Sharma-Tasso-Olver (STO) equation (see example \ref{STO}). The image of $\gen w$ is $\gen \phi$ under $\phi=w_x$. The image of $e^{-w} \varrho(t,x)\gen w$ is more complicated (which depends on the nonlocal variable $w=\partial_x^{-1}\phi$ as well) and obtained by calculating its first prolongation and then replacing $w_x$ by $\phi$
\begin{equation}\label{nlocal}
  \tilde{\mathbf{v}}(\varrho)=-e^w(\varrho \phi-\varrho_x)\gen \phi.
\end{equation}
This provides an infinite-dimensional nonlocal symmetry of the STO equation which is invariant under only a finite dimensional symmetry algebra with basis $\langle\gen t, 3t\gen t+x\gen x, \gen \phi\rangle$.

\par II.) The second case for $\beta$ is $\beta\not=\alpha(\alpha-1)$. We restrict to the case $a=0$. This is a non-linearizable case (Lie point symmetry algebra is finite-dimensional).
First we note that the transformation $u=w^m$ does not change the form of the equation, but it maps the parameters as follows
\begin{equation}\label{shifted-param}
  \alpha\to \tilde{\alpha}=m(\alpha+1)-1,  \quad \beta\to \tilde{\beta}=(3\alpha+\beta+1)m^2-3(\alpha+1)m+2.
\end{equation}
The pairs of parameters $(\alpha,\beta)=(-1,2)$ for any $m$ and $(\alpha,\beta)=(-1,\beta)$, $\beta\ne 2$ for $m=-1$  remain unchanged under this transformation, which hence defines a discrete transformation in these special cases. One can make one of the parameters zero by an appropriate choice of $m$.
Choosing $m=\pm[2/(2-\beta)]^{1/2}$  we find that the case $\alpha=-1$, $\beta\ne 2$ is equivalent to  $\alpha=-1$, $\beta=0$. This means that we can put $\beta=0$ without loss of generality in \eqref{special-pde} and obtain
\begin{equation}\label{beta-not-2}
  u_t=u_{xxx}-3\frac{u_x u_{xx}}{u}.
\end{equation}
Now, expressing \eqref{beta-not-2}
in terms of the differential invariant (using the scale invariance in $u$) $w=u_x/u$ (which is actually a Cole-Hopf transformation)
maps the equation to the modified KdV (mKdV) equation
$$w_t=w_{xxx}-6w^2w_x.$$  It is well-known that the Miura transformation $\psi=w^2+ w_x$ maps $w_t=w_{xxx}-6w^2w_x$ to the KdV equation $\psi_t=\psi_{xxx}-6ww_x$. The final transformation mapping \eqref{beta-not-2}
to the  KdV equation is given by
$\psi=u_{xx}/{u}$.
In a recent paper \cite{Carillo2017}, equation \eqref{beta-not-2} has been presented as a novel KdV-type third order evolution equation. In this work, the the Cole-Hopf link between  Eq. \eqref{beta-not-2} and the mKdV equation has been reestablished and this fact has also been  used as an alternative method to prove that a recursion operator exists for   \eqref{beta-not-2}. Also, a nonlocal invariance of \eqref{beta-not-2}, not treated in the present paper, has been constructed using directly the composition of the M\"obius transformation \eqref{Mobius} with Baecklund transformations  connecting \eqref{beta-not-2} to the Schwarzian-KdV  equation \eqref{SKdV} (see the last paragraph of Example \eqref{non-lin}).

When $\alpha\ne -1$ we choose $m=(\alpha+1)^{-1}$ to make $\tilde{\alpha}=0$ in which case $\tilde{\beta}=(\alpha+1)^{-1}[\beta-\alpha(\alpha-1)]$. This implies that the case $\alpha\ne -1$ is equivalent to $\alpha=0$. When $\alpha\ne -1$, $\beta\ne 2$ we can always set $\beta\to 0$. In the former case, by the substitution $w=u_x/u$ the corresponding equation
$$u_t=u_{xxx}+\beta u^{-2}u_x^3$$ is reduced  to
$$w_t=w_{xxx}+\frac{3}{2}(w^2)_{xx}+(\beta+1)(w^3)_x.$$ This is the STO equation for $\beta=0$ which establishes once again its connection with the linear KdV equation (see also example \ref{STO}).

The equation
\begin{equation}\label{special-pde-c}
  u_t=u_{xxx}+c\left(2\frac{u_x u_{xx}}{u}-\frac{u_x^3}{u^2}\right).
\end{equation}
is non-linearizable for arbitrary $c\neq 0, -3/4$:
the linearity condition is satisfied only for $c=-3/4$ ($\alpha=-1/2$). The linearizing transformation is $w=\sqrt{u}$.
For $c=-3/2$ ($\alpha=-1$, $\beta=3/2$) it is equivalent to \eqref{beta-not-2} by the substitution $u=w^2$ (see example \ref{skdv-ex}). In other words, the corresponding equation can be mapped to mKdV.

Although equation \eqref{class-f-g} is non-linearizable when $g\ne f'+f^2$, if it is invariant under translations in $u$ (so $f, g$ are independent of $u$) or in general admits at least one symmetry of the form $X=\xi(t,x,u)\gen x+\phi(t,x,u)\gen u$, then the equation can be  differentiated with respect to $x$ so that the (non-point) transformation $v=u_x$ (a differential invariant of the translational symmetry $\gen u$) takes it to one which can either be linearized or mapped to the KdV equation.

\begin{exa}
\begin{equation}\label{ex-c}
  u_t=u_{xxx}+c u_x^{-1}u_{xx}^2.
\end{equation}
This equation admits a finite dimensional maximal symmetry algebra spanned by
$$\langle\gen t, \gen x, \gen u, t\gen t+\frac{x}{3}\gen x, u\gen u\rangle.$$
Since equation \eqref{ex-c} is invariant under $\gen u$, the differential transformation $(t,x,u)\to (t,x,u_x)$, with $v=u_x$,   transforms it to the form
$$v_t=v_{xxx}+c\left(2\frac{v_x v_{xx}}{v}-\frac{v_x^3}{v^2}\right),$$
which is exactly \eqref{special-pde-c} for $c=-3/4$ (the linearizable case).
In other words, it can admit an infinite-dimensional point symmetry group
$$\mathbf{v}(\gamma)=\gamma(t,x)\sqrt{v}\gen v, \quad \gamma_t=\gamma_{xxx}$$
in addition to the finite-dimensional symmetry algebra
$$\langle\gen t, \gen x, \gen v, t\gen t+\frac{x}{3}\gen x, v\gen v\rangle.$$
So the transformation $w=\sqrt{v}$ taking $\mathbf{v}(\gamma)$ to $\gamma\gen w$ linearizes
$$v_t=v_{xxx}-\frac{3}{4}\left(2\frac{v_x v_{xx}}{v}-\frac{v_x^3}{v^2}\right).$$

Finally, the nonlocal transformation $w=\sqrt{u_x}$ maps
\begin{subequations}\label{equi-KdV}
  \begin{equation}
  u_t=u_{xxx}-\frac{3}{4} u_x^{-1}u_{xx}^2
\end{equation}
to
\begin{equation}
  w_t=w_{xxx}.
\end{equation}
\end{subequations}

\end{exa}

\begin{rmk}
  A more general approach to producing nonlinear evolution equations having nonlocal or more precisely quasi-local symmetries from those equations admitting Lie point symmetries of the form
$$X=\xi(t,x,u)\gen x+\phi(t,x,u)\gen u$$ can be found in \cite{Zhdanov2010}. Note that we already know that for any third-order evolution equation from our class (see for example \cite{Basarab-HorwathGuengoerLahno2013}) this vector field can be mapped to the normal form $\gen u$. In other words, a point transformation exists taking the equation to one which is independent of the dependent variable $u$. So a sequence of operations reducing vector fields to normal form and applying the substitution $(t,x,u)\to (t,x,u_x)$  in a finite number of steps may lead to equations with nonlocal (or quasi-local) symmetries.
\end{rmk}

\begin{exa}\label{skdv-ex}
The integrable Schwarzian-KdV (SKdV) equation
\begin{equation}\label{SKdV}
  u_t=u_{xxx}+c u_x^{-1}u_{xx}^2,  \quad c=-\frac{3}{2}
\end{equation}
admits a six-dimensional symmetry algebra with the structure
$\Sl(2,\mathbb{R})\oplus \sf{A}_3$ (see for example
\cite{GungorLahnoZhdanov2004}).

This equation was discovered by J. Weiss \cite{Weiss1983} (a subcase of the Krichever-Novikov equation) as the singular manifold equation for the KdV and mKdV equations.

The symmetry group of the SKdV equation is composed of the M\"obius transformations of the dependent variable $u$
\begin{equation}\label{Mobius}
  u'=\frac{au+b}{cu+d}, \quad ad-bc\ne 0,
\end{equation}
and simultaneous rescaling of $x$ and $t$ and  translations in $x$ and $t$: $x'=\alpha x+\beta$,  $t'=\alpha^3 t+\gamma,$ $\alpha>0$.

Applying the transformation $v=u_x$ to the equation gives
$$v_t=v_{xxx}-\frac{3}{2}\left(2\frac{v_x v_{xx}}{v}-\frac{v_x^3}{v^2}\right),\quad v=u_x,$$ which is not linearizable as is clear from \eqref{special-pde-c} and the subsequent paragraph. Inspired by the previous arguments we can use
the scaling differential invariant  $w=v_x/v=u_{xx}/u_x$ to transform it to the MKdV equation
$$w_t=w_{xxx}-\frac{3}{2}w^2w_x.$$
The Miura transformation $\tilde{w}=1/4w^2\pm 1/2 w_x$ which maps this equation to the  KdV equation
$$\tilde{w}_t=\tilde{w}_{xxx}-6\tilde{w}\tilde{w}_x$$
produces the Schwarzian transformation (the unique M\"obius differential invariant) between the Schwarzian-KdV and KdV equation:
$$\tilde{w}=-\frac{1}{2}\{u;x\}=-\frac{1}{2}\left[\frac{u_{xxx}}{u_x}-\frac{3}{2}\left(\frac{u_{xx}}{u_x}\right)^2\right].$$

There is another transformation that gives the same result, namely:
$$\tilde{w}=\frac{1}{2}\left[\frac{u_{xxx}}{u_x}-\frac{1}{2}\left(\frac{u_{xx}}{u_x}\right)^2\right]
=\frac{(\sqrt{u_x})_{xx}}{\sqrt{u_x}}.$$

The B\"acklund transformation establishing this connection is well-known (see for example \cite{Hone2009}):
If $w$ and $\tilde{w}$ are both solutions of the KdV equation
$$w_t=w_{xxx}-6ww_x,$$ then
\begin{equation}\label{Baeck}
  w=2(\log u)_{xx}+\tilde{w}
\end{equation}
is a B\"acklund transformation for the KdV equation provided that $u$ satisfies
$$\frac{u_t}{u_x}=\{u;x\}+k^2,  \quad \tilde{w}=6k^2+\frac{(\sqrt{u_x})_{xx}}{\sqrt{u_x}}.$$
The transformation \eqref{Baeck} can be derived from the truncation of the Laurent-type expansion for the   KdV equation at the zero order term.
We note that $k^2$ can be removed by a Galilien transformation.
\end{exa}

Using our approach we want to consider the Harry-Dym equation and establish the
connection with the Schwarzian-KdV equation.

\begin{exa}

The Harry-Dym equation
\begin{equation}\label{harry}
  u_t=u^3u_{xxx}
\end{equation}
is integrable and hence possesses an infinite hierarchy of higher (generalized) symmetries and conservation laws \cite{LeoLeoSolianiSolombrino1983, MikhailovSokolov2009}.
The point symmetry algebra has the structure $\lie=\Sl(2,\mathbb{R})\oplus\sf{A}_2$ \cite{Basarab-HorwathGuengoerLahno2013}. This equation admits peakons and compactons (i.e. weak soliton solutions).

By using the multiplier (or characteristic function) $\mu(u)=u^{-2}$ and noting that
$$uu_{xxx}=D_x[uu_{xx}-\frac{1}{2} u_x^2]$$ we can re-express \eqref{harry}  in the form of a conservation law:

$$v_t=D_x\left[\frac{v_{xx}}{v^3}-\frac{3}{2}\frac{v_x^2}{v^4}\right],  \quad v=u^{-1}.$$
Using the substitution $v=w_x$ or $w=D_x^{-1}v$ (inverse differentiation or potentiation transformation) we find the potential equation
\begin{equation}\label{potential}
  w_t=\frac{w_{xxx}}{w_x^3}-\frac{3}{2}\frac{w_{xx}^2}{w_x^4}=w_x^{-2}\{w;x\},
\end{equation}
for which the point symmetry algebra is six-dimensional with a basis \cite{Basarab-HorwathGuengoerLahno2013}
\begin{equation}\label{sym-SKdV}
  \langle \gen x, 2x\gen x, x^2\gen x\rangle \oplus \langle \gen t, \gen w, t\gen t+\frac{w}{3}\gen w\rangle.
\end{equation}
This algebra is isomorphic to that of the SKdV equation. Note that $w=D_x^{-1}u^{-1}$ is a nonlocal variable so that it induces a nonlocal symmetry.
The hodograph transformation
$$t=t, \quad y=w(t,x), \quad z(t,y)=x$$
implies the following transformation rules
$$w_t=-\frac{z_t}{z_y}, \quad \{w;x\}=-w_x^2\{z;y\}$$
and transforms \eqref{potential} to the SKdV equation
\begin{equation}\label{skdv-eq}
  \frac{z_t}{z_y}=\{z;y\}.
\end{equation}
From the previous example it follows that the differential substitution $z=\{\tilde{z};y\}$ (differential invariant) maps it to the KdV equation
$$\tilde{z}_t=\tilde{z}_{yyy}+3\tilde{z}\tilde{z}_y.$$

\end{exa}

Another equation which can be written in conservation form is given in the following example:
\begin{exa}\label{STO-eq}
Sharma-Tasso-Olver (STO) equation (a conservation law):
\begin{equation}\label{STO}
  u_t=u_{xxx}+3uu_{xx}+3u_x^2+3u^2u_x=u_{xxx}+\frac{3}{2}(u^2)_{xx}+(u^3)_x
\end{equation}
admits a finite-dimensional (3-dimensional) point symmetry algebra (compare this with Eq. \eqref{gen-STO} for $\alpha=a=0$).

However, when written as a system by  introducing the auxiliary variable $v$ (the potential function), we have
\begin{equation}\label{sys}
 v_x=u,\quad v_t=u_{xx}+u^3+3uu_x,
\end{equation}
which admits an infinite-dimensional  symmetry algebra generated by the vector fields
\begin{equation}\label{sym-STO}
  \begin{split}
      & T=\gen t,\quad P=\gen x,\quad D=x\gen x+3t\gen t-u\gen u, \\\
      & \mathbf{v}(\varrho)=e^{-v}[(\varrho u-\varrho_x)\gen u-\varrho \gen v],\quad \text{with}\;\;  \varrho_t=\varrho_{xxx}
   \end{split}
\end{equation}
\end{exa}

Projecting the symmetry $\mathbf{v}(\varrho)$ onto the $(t,x,v)$ coordinates suggests that if we put  $v=\ln \psi$, the equation is changed to $\varrho \gen \psi$ so that $u=v_x=\psi_x/\psi$.
This immediately shows that there is a non-point transformation $u=\psi_x/\psi$ (the Cole-Hopf transformation) mapping the STO equation to
$$D_x\left(\frac{\psi_t-\psi_{xxx}}{\psi}\right)=0.$$
So whenever $\psi$ solves the linear KdV equation
$$\psi_t=\psi_{xxx},$$ then $u=\psi_x/\psi$ solves the STO equation.

Alternatively, we can look at the symmetry algebra of the potential equation obtained by eliminating $u$ in \eqref{sys}
\begin{equation}\label{potential-2}
  v_t=v_{xxx}+3v_xv_{xx}+v_x^3.
\end{equation}
This equation is still in the class of linearizable equations \eqref{w-invar-pde}. It has an infinite dimensional symmetry algebra with a basis given by $\langle\gen t, \gen x, 3t\gen t+x\gen x, \gen v, \mathbf{v}(\varrho)\rangle$, where
$$\mathbf{v}(\varrho)=e^{-v}\varrho\gen v,  \quad \varrho_t=\varrho_{xxx}$$
from the formula \eqref{inf-span} for $\alpha=0$.

The transformation $v=\ln \psi$  linearizes \eqref{potential-2} to $\psi_t=\psi_{xxx}$ with symmetry $\tilde{\mathbf{v}}(\varrho)=\varrho\gen{\psi}$ reflecting the superposition law for the linear PDE $\varrho_t=\varrho_{xxx}$.

\begin{rmk}
  As a by-product we have obtained  an infinite-dimensional nonlocal symmetry of the original equation in the form (see  \eqref{nlocal})
$$\hat{\mathbf{v}}(\varrho)=e^{-v}(\varrho u-\varrho_x)\gen u,  \quad v=\partial_x^{-1}u.$$
\end{rmk}

\subsection{Conservation laws}\label{con-law}
We wish to construct local conservation laws for the class \eqref{class-f-g} by looking for characteristic functions (or multipliers) of the form $Q(u)$. We require the product $Q\Delta$, where $\Delta=u_t-u_{xxx}-F$, where $F$ is of the form \eqref{class},  to be a total time-space divergence in the form
\begin{equation}\label{CL}
  Q\Delta=D_t T+D_x X=0,
\end{equation}
on all solutions $u(t,x)$. Here $D_t$, $D_x$ are total differential operators with respect to $t$ and $x$ respectively and  $T$, $X$ are conserved density and flux (differential functions to be determined by $Q$), respectively.
In view of \eqref{CL}, the spatial integral of  the conserved density $T$ satisfies
$$\frac{d}{dt}\int_{\mathbb{R}}Tdx=-X\Big\vert_{-\infty}^\infty.$$ If the flux $X$ vanishes at $x\to\pm \infty$, then
$$\mathcal{T}[u]=\int_{\mathbb{R}}Tdx=\rm{const.}$$ is a conserved quantity for $\Delta=0$.

A necessary and sufficient condition for the existence of the characteristic function $Q$ is that we must have
\begin{equation}\label{CL-cond}
  \mathsf{E}_u[Q\Delta]=0,
\end{equation}
where $\mathsf{E}_u$ is the temporal-spatial Euler-Lagrange operator with respect to $u$ defined by
$$\mathsf{E}_u=\frac{\partial}{\partial u}-D_t\frac{\partial}{\partial u_t}-D_x\frac{\partial}{\partial u_x}+D_x^2\frac{\partial}{\partial u_{xx}}-D_x^3\frac{\partial}{\partial u_{xxx}}.$$
Realizing this condition and then splitting with respect to the linearly independent derivatives we obtain determining equations for $Q$. This also gives conditions on $f$.

We distinguish two cases:

\par 1.) $a\ne 0$, $g$ {\it a priori} generic: We obtain a system of two determining equations and find $Q'=fQ$. This automatically requires the linearity condition $f'+f^2=g$ and means that the conservation law is exactly the linear KdV equation
\begin{equation}\label{CL-1}
  D_t\left[\int Q(u)du\right]=D_x^3\left[\int Q(u)du\right]+a D_x^2\left[\int Q(u)du\right],  \quad Q(u)=e^{\int f(u)du}.
\end{equation}
We recover the linearizing transformation $w=\int Q(u)du$ (a conserved density).

\par 2.) $a=0$, $g$ {\it a priori} generic: The determining equations give a second order linear equation for $Q$:
\begin{equation}\label{a-0}
  Q''-3fQ'+(2g-3f')Q=0.
\end{equation}
In this case we can find  two linearly independent, nonconstant characteristic functions $Q$ when $2g-3f'\ne 0$. Again, one can see that if we impose the condition $Q'=fQ$ on \eqref{a-0} then we must have the linearity condition $g=f'+f^2$. This implies that even in the linearizable case we can obtain two different characteristic functions. Otherwise we have to find at least one solution of \eqref{a-0} by inspection or by known standard methods of solving a variable coefficient second order linear equation. If we have $2g-3f'= 0$ then one solution is at our disposal by just one quadrature.

In the linearizable case, with the condition $2g-3f'\ne 0$,  we have
\begin{equation}\label{Q1-2}
  Q_1=e^{\int f(u)du},  \qquad Q_2=Q_1 \int Q_1(u)du.
\end{equation}
The first solution leads to the conservation law \eqref{CL-1} with $a=0$. From the second solution we find the conservation law (not a linear KdV equation)
\begin{equation}\label{CL-2}
  D_t \left[\int Q_2 du\right]=D_x\left[Q_2 u_{xx}+\frac{1}{2}(3fQ_2-Q_2')u_x^2\right].
\end{equation}
We can write \eqref{CL-2} as a potential system
\begin{equation}\label{pot-sys}
  v_x=\int Q_2 du,  \qquad v_t=Q_2 u_{xx}+\frac{1}{2}(3fQ_2-Q_2')u_x^2.
\end{equation}
Of course, if we can solve the first equation of \eqref{pot-sys} for $u$, we can substitute it into the second equation and obtain the potential equation for $v$. This potentiation process preserves the form of the original equation.

In the special case $g=0$ ($f'\ne 0$) the left-hand side of \eqref{a-0} is an exact differential and can be integrated to give the characteristic functions $Q$:
\begin{equation}\label{Q}
  Q_1^3,  \qquad Q_1^3 \int  Q_1^{-3}du,  \quad \text{where}\;\; Q_1=e^{\int f(u)du}.
\end{equation}
for the conservation laws of
\begin{equation}\label{g=0}
  u_t=u_{xxx}+3f(u)u_x u_{xx}.
\end{equation}

For the special case $f(u)=\alpha u^{-1}$ the linearization (for $\alpha=1$, $f'+f^2=0$), transformation to the KdV equation (for $\alpha=-1$) and conservation laws for different values of $\alpha$ have recently been investigated in \cite{SenAhalparaThyagarajaKrishnaswami2012, SilvaFreireSampaio2016}. The equation  studied there depends on two parameters $a, \epsilon$, one of which can be  rescaled to be equal to $1$. Let us remark that Eq. \eqref{g=0} admits an infinite-dimensional symmetry group indicating its equivalence to a linear equation. When $\alpha=-1$, Eq. \eqref{g=0} is only invariant under a finite-dimensional (four-dimensional) symmetry group. In this case, the Cole-Hopf transformation (differential invariant of the symmetry $u\gen u$) maps its to the MKdV equation and hence to the KdV equation by the Miura transformation.

We now present some examples of subcases of \eqref{class-f-g}.

\begin{exa}
We consider the linearizable subclass of \eqref{special-pde} for the special choice $a=0$
\begin{equation}
 u_t=u_{xxx}+\frac{3\alpha}{u}u_xu_{xx}+\frac{\alpha(\alpha-1)}{u^2}u_x^3.
\end{equation}
Here we have $f=\alpha u^{-1}$ so that
$$Q_1(u)=u^{\alpha},  \qquad Q_2(u)=\frac{1}{\alpha+1}u^{2\alpha+1}, \quad \alpha\ne -1.$$
We are interested in the second solution as we know the first solution produces  the linearizing transformation $w=u^{\alpha+1}$. From \eqref{CL-2} we obtain the conservation law
$$D_t[u^{2(\alpha+1)}]=2(\alpha+1)D_x\left[u^{2\alpha+1}u_{xx}+\frac{\alpha-1}{2}u^{2\alpha}u_x^2\right].$$
We define $v_x=u^{2(\alpha+1)}$ for potentiation. Substituting $u=v_x^{1/2(\alpha+1)}$ into
$$v_t=2(\alpha+1)u^{2\alpha+1}u_{xx}+\frac{\alpha-1}{2}u^{2\alpha}u_x^2$$ yields the potential equation for $v$
\begin{equation}\label{potential-3}
  v_t=v_{xxx}-\frac{3}{4}v_x^{-1}v_{xx}^2.
\end{equation}
From example \eqref{ex-c} we know that this equation is linearizable to the linear KdV equation by the substitution $w=\sqrt{v_x}=u^{\alpha+1}$. We have shown that   producing a conservation law by $Q_2$ is equivalent to the linearization of the equation (as well when a potential form of the equation is introduced from the conservation law).

For $\alpha=-1$ we have
$$Q_1=u^{-1},  \quad Q_2=\frac{\ln u}{u}.$$
$Q_1$ corresponds to the linearization $D_t(\ln u)=D_x^3 (\ln u)$. For $Q_2$, we can, using \eqref{CL-2}, write the conservation law
$$D_t[(\ln u)^2]=D_x\left[2\frac{\ln u}{u}-\left(\frac{2\ln u+1}{u^2}\right)u_x^2\right].$$
The substitution $u=e^{\sqrt{v_x}}$  transforms  it to \eqref{potential-3} which is linearized by the transformation $w=\sqrt{v_x}$. The final transformation linearizing the initial equation is $w=\ln u$ (which corresponds to $Q_1$).
\end{exa}

\begin{exa}\label{non-lin}
  An equation that is nonlinearizable  by a point transformation:
\begin{equation}
  u_t=u_{xxx}+c\left(2\frac{u_x u_{xx}}{u}-\frac{u_x^3}{u^2}\right),  \quad c\not=-\frac{3}{4}.
\end{equation}
We have $2g-3f'=0$ and from Eq. \eqref{a-0} we find only one nonconstant $Q(u)=u^{2c+1}$, $c\not=-1/2$. For $c=-1/2$ we have $Q(u)=\ln u$. The corresponding conservation law for the first case is
$$D_t[u^{2(c+1)}]=2(c+1)D_x[u^{2c+1}u_{xx}-\frac{1}{2}u^{2c}u_x^2],  \quad c\not= -1,-\frac{1}{2}.$$
If we define $u=\phi_x^{1/(2(c+1))}$ and integrate the above conservation law we get the potential equation for $\phi$
\begin{equation}\label{potential-4}
  \phi_t=\phi_{xxx}-\frac{4c+3}{4(c+1)}\frac{\phi_{xx}^2}{\phi_x}.
\end{equation}
When $c=-3/4$ the equation is linearizable by the differential substitution $u=\phi_x^2$. We note that we can always absorb any arbitrary integration  function of $t$ into $\phi$.

In the special case $c=-3/2$, Eq. \eqref{potential-4} reduces to the SKdV equation.  In other words,
\begin{equation}\label{int-sol-KdV}
  u_t=u_{xxx}-\frac{3}{u}u_xu_{xx}+\frac{3}{2u^2}u_x^3
\end{equation}
is mapped to a differential consequence of the SKdV in $\phi$ : $D_x[SKdV]=0$  by the substitution $u=\phi_x^{-1}$. Invariance of \eqref{int-sol-KdV} under $u\to u^{-1}$ produces another substitution $u=\phi_x$ achieving the same connection.

For $c=0$ we  recover the link \eqref{equi-KdV}.

As a final remark, we make the observation that Eq. \eqref{beta-not-2}, locally equivalent to \eqref{int-sol-KdV}, can be written in the conservative form
$$2uu_t=D_t(u^2)=2D_x[uu_{xx}-2u_x^2].$$
The differential substitution (or Baecklund transformation in the terminology of \cite{Carillo2017}) $u^2=\phi_x$ maps it to the SKdV equation in $\phi$ (see \eqref{SKdV}) after integration with respect to $x$. This fact is presented  in \cite{Carillo2017} as Proposition 2. The fact that that Eq. \eqref{beta-not-2} is invariant under the reciprocal transformation $u\to u^{-1}$ demonstrates that there is another Baecklund transformation $u^{-2}=\phi_x$ realizing the above-mentioned link. The latter Baecklund transformation can also be derived using the second multiplier $Q_2=u^{-3}$ provided by the formula \eqref{Q} for $f(u)=-1/u$. The corresponding conservation law is
$$D_t(u^{-2})=-2D_x(u^{-3}u_{xx}).$$
\end{exa}


\end{document}